\documentclass[reqno,11.5pt]{article}

\usepackage{diagrams}
\usepackage{amssymb}
\usepackage{tikz}
\usepackage{amsmath}
\usepackage{pdfpages}
\usetikzlibrary{positioning}
\usepackage{mathtools}
\usepackage{environ}
\usepackage{latexsym}
\usepackage{mathrsfs}
\usepackage{amsthm}
\usepackage{verbatim}
\usepackage{graphicx}
\usepackage{tikz-cd}
\usepackage{epstopdf}
\usepackage{epsfig}
\usepackage{color}
\usepackage{bm}
\usepackage[colorlinks,linkcolor=blue,anchorcolor=green,citecolor=red]{hyperref}
\usepackage{amscd}
\usepackage{tikz-cd}

\newcounter{quote}

\NewEnviron{myquote}{\vspace{1ex}\par
\refstepcounter{quote}%
\hfill\parbox{\dimexpr \textwidth-2cm}%
{\centering\small\textit{\BODY}}%
\hfill\llap{(\thequote)}\vspace{1ex}\par}

\newtheorem{theorem}{Theorem}
\newtheorem{remark}{Remark}

\newtheorem{corollary}{Corollary}
\newtheorem{lemma}{Lemma}

\newcommand\skw{\operatorname{skw}}
\newcommand\Skw{\operatorname{Skw}}
\newcommand\sym{\operatorname{sym}}
\newcommand\grad{\operatorname{grad}}
\renewcommand\div{\operatorname{div}}
\newcommand\curl{\operatorname{curl}}
\newcommand\rot{\operatorname{rot}}

\newcommand\inc{\operatorname{inc}}
\newcommand\air{\operatorname{airy}}
\newcommand\deff{\operatorname{def}}

\newcommand\R{\mathbb{R}}

\makeatletter
\@addtoreset{equation}{section}
\makeatother

\usepackage{geometry}
\begin{document}

\title{Poincar\'{e} path integrals for elasticity}


\author{Snorre H. Christiansen\thanks{Department of Mathematics, University of Oslo, PO Box 1053 Blindern, NO 0316 Oslo, Norway. 
              email:{\tt snorrec@math.uio.no}} \and Kaibo Hu\thanks{Corresponding author. School of Mathematics, University of Minnesota,
           206 Church St. SE, Minneapolis, MN, USA.  email: {\tt khu@umn.edu}}\and Espen Sande\thanks{Department of Mathematics, University of Oslo, PO Box 1053 Blindern, NO 0316 Oslo, Norway.
              email:{\tt espsand@math.uio.no}}}

\date{}

\maketitle

\begin{abstract}
We propose a general strategy to derive null-homotopy operators for differential complexes based on the Bernstein-Gelfand-Gelfand (BGG) construction and properties of the de Rham complex. Focusing on the elasticity complex, we derive path integral operators $\mathscr{P}$ for elasticity satisfying $\mathscr{D}\mathscr{P}+\mathscr{P}\mathscr{D}=\mathrm{id}$ and $\mathscr{P}^{2}=0$, where the differential operators $\mathscr{D}$ correspond to the linearized strain, the linearized curvature and the divergence, respectively. {In general we derive path integral formulas in the presence of defects.} As a special case, this gives the classical Ces\`{a}ro-Volterra path integral for strain tensors satisfying the Saint-Venant compatibility condition. 
\end{abstract}
   
\noindent Keyworlds: homotopy operator, Ces\`{a}ro-Volterra path integral, Bernstein-Gelfand-Gelfand resolution, elasticity, defect

\section{Introduction}
Let $\Lambda^{k}(\Omega)$  be the space of smooth differential $k$-forms on an open domain $\Omega\subset\R^n$. The de Rham complex then reads
\begin{equation}\label{sequence:deRham}
\begin{diagram}
0& \rTo &\mathbb{R} & \rTo & \Lambda^{0}(\Omega) & \rTo^{d_{0}} & \Lambda^{1}(\Omega) & \rTo^{d_{1}} &  \cdots & \rTo^{d_{n-1}} &   \Lambda^{n}(\Omega) & \rTo & 0,
\end{diagram}
\end{equation}
where $d_{k}$, the $k$th exterior derivative, satisfies $d_kd_{k-1}=0$. In three space dimensions, $d_0$ corresponds to the gradient operator, $d_1$ corresponds to the curl and $d_2$ corresponds to the divergence.
It is well known that for the de Rham complex on a contractible domain, there exist null-homotopies, or Poincar\'e operators $\mathfrak{p}_k:\Lambda^{k}(\Omega)\to\Lambda^{k-1}(\Omega)$,
which satisfy
\begin{equation}\label{eq:nullhom}
\mathfrak{p}_{k+1}d_{k}+d_{k-1}\mathfrak{p}_k=\mathrm{id}_{\Lambda^{k}(\Omega)}.
\end{equation}
When it is clear from context it is common to drop the indices on both the exterior derivatives $d$ and the Poincar\'e operators $\mathfrak{p}$.
The existence of the  Poincar\'e operators implies the Poincar\'{e} lemma, i.e., that for any $k$-form $\omega$ satisfying $d_{k}\omega=0$, there exists, locally, a $(k-1)$-form $\phi$ such that $d_{k-1}\phi=\omega$. Using \eqref{eq:nullhom} we see that a choice of $\phi$ is $\phi=\mathfrak{p}_{k}\omega$.
In addition to being null-homotopies, the Poincar\'e operators also satisfy
 \begin{itemize}
 \item [(i)]
 the complex property: $\mathfrak{p}^2=\mathfrak{p}_{k}\mathfrak{p}_{k+1}=0$;
 \item [(ii)]
the polynomial preserving property: if $\omega$ is a homogeneous polynomial of degree $r$, then $\mathfrak{p}_k\omega$ is  a homogeneous polynomial of degree $r+1$. 
 \end{itemize}
The polynomial preserving property reflects the fact that the differential operators in the de Rham complexes are homogeneous first order operators.  

Due to the complex property, $\mathfrak{p}\phi=0$ provides a gauge condition for a potential $\phi$ in the following sense. For any $\omega\in \Lambda^{k}(\Omega)$ with $d_{k}\omega=0$, a potential $\phi\in \Lambda^{k-1}(\Omega)$ satisfying both $\mathfrak{p}_{k-1}\phi=0$ and $d_{k-1}\phi=\omega$, is uniquely determined and given by $\phi=\mathfrak{p}_{k}\omega$.

Furthermore the operators $\mathfrak{p}$ can be given an explicit representation in terms of path integrals, which has been important for many applications.  
Using these path integrals one can obtain the Koszul operators, a main tool in the construction of finite elements for scalar and vector field problems \cite{Arnold.D;Falk.R;Winther.R.2006a,hiptmair1999canonical}.
By averaging the base point of the Poincar\'{e} operators, Costabel and McIntosh \cite{costabel2010bogovskiui} constructed {Bogovski\u\i} type operators which they used to prove regularity results for the de Rham complex in Sobolev spaces. This leads to some very useful inequalities with applications in the analysis of electromagnetic problems and finite element methods (see, e.g., \cite{boffi2011discrete,bonito2013regularity}).  The homotopy identity, polynomial-preserving property and the complex property are important in these applications.

Let $\mathbb{V}$ and $\mathbb{S}$ be the linear space of vectors and symmetric matrices in three space dimensions and let $C^{\infty} (\Omega; \mathbb{V})$ and $C^{\infty} (\Omega; \mathbb{S})$ denote, respectively, the spaces of smooth vector- and symmetric-matrix-valued functions.
The linear elasticity complex in three space dimensions reads
\begin{align}\label{sequence:3Delasticity}
 \begin{diagram}
\mathrm{0} & \rTo^{}  &\mathrm{RM} & \rTo^{{\subseteq}} &  C^{\infty}(\Omega; \mathbb{V}) & \rTo^{\deff} &  C^{\infty}(\Omega;  \mathbb{S})  & \rTo^{\inc} &  C^{\infty}(\Omega; \mathbb{S}) & \rTo^{\div} &  C^{\infty}(\Omega; \mathbb{V}) & \rTo & 0,
 \end{diagram}
\end{align} 
with the differential operators in the vector form and index form
\begin{equation*}
\begin{aligned}
\deff {u}&:=\frac{1}{2}\left ( \nabla {u}+ {u}\nabla\right ),
&&(\deff {u})_{ij}=\frac{1}{2}\left (\partial_{i}{u}_{j}+\partial_{j}{u}_{i}\right ), 
&&&{u}\in C^{\infty}(\Omega; \mathbb{V}),
\\
\inc E &:= \nabla\times E\times \nabla, 
&&(\inc E)_{ij}=\epsilon_{ist}\epsilon_{jlm}\partial^{s}\partial^{l}E^{tm}, 
&&&E\in C^{\infty}(\Omega; \mathbb{S}), 
\\
\div V&:=\nabla\cdot V, 
&&(\div V)_{i}=\partial^{j}V_{ij}, 
&&&V\in C^{\infty}(\Omega; \mathbb{S}).
\end{aligned}
\end{equation*}
Here $\epsilon$ is the permutation tensor. The kernel of the linearized deformation operator $\deff$, i.e., $\mathrm{RM}:=\{{u}={a}+{b}\wedge {x}: {a}, {b}\in \mathbb{V}\}$, is called the space of rigid body motions. Given ${u}\in C^{\infty}(\Omega; \mathbb{V})$,  $\deff {u}$ is the symmetric gradient or (linearized) deformation \cite[p.~149]{Taylor2010}. Given $E\in  C^{\infty}(\Omega;  \mathbb{S}) $,  $\inc E :=  \nabla\times E\times \nabla$ is called the incompatibility of the strain (metric) tensor $E$, where $\nabla\times$ and $\times \nabla$ denote, respectively, the column-wise curl and the row-wise curl of a matrix field.

Kr\"{o}ner is one of the pioneers of relating the incompatibility of the strain tensor with defect densities of the material \cite{kroner1963dislocation, kroner1981continuum, van2010non}, and therefore \eqref{sequence:3Delasticity} is also referred to as the Kr\"{o}ner complex in the literature. We also refer to \cite{amstutz2016analysis,amstutz2017incompatibility} for the analysis and modeling of defects with the $\inc$ operator and to \cite{angoshtari2015differential,hackl1988existence,khavkine2017calabi} for applications of differential complexes in elasticity and geometry.
The elasticity complex \eqref{sequence:3Delasticity} has also been used to construct stable finite elements for the Hellinger-Reissner formulation of elasticity \cite[p. 121]{Arnold.D;Falk.R;Winther.R.2006a}.

Comparing the two complexes \eqref{sequence:deRham} and \eqref{sequence:3Delasticity} there are now two natural questions to ask:
\begin{itemize}
\item Does there exist Poincar\'e operators for the elasticity complex that satisfy a null-homotopy relation (analogous to \eqref{eq:nullhom}), the complex property and a polynomial preserving property?
\item If so, what are the explicit formulas, as path integrals, for them?
\end{itemize}
The main result of this paper is to provide a positive answer to both of these questions. Our approach is to use the Poincar\'e path integrals for the de Rham complex together with the Bernstein-Gelfand-Gelfand (BGG) resolution, a general construction that can be used to derive the elasticity complex from the de Rham complex \cite{arnold2006defferential,eastwood2000complex,falk2008finite}. We then obtain Poincar\'e path integrals for the elasticity complex.

We remark that of the three Poincar\'e path integrals for the elasticity complex, the first is already known: this is a result in the classical theory of linear elasticity that dates back to the work of Ces\`{a}ro in 1906 and Volterra in 1907
 \cite{cesaro1906sulle, volterra1907equilibre}. The two other Poincar\'e path integrals we derive, and that together provide the full sequence of null-homotopies, appear to be new. Recall that the symmetric strain tensor $E$ in elasticity satisfies the Saint Venant compatibility condition $\inc E = 0$ and one can then show that on a contractible domain $\Omega$, $E$ is the deformation of some displacement vector field ${u}$, i.e. $E=\deff {u}$.
Moreover, the displacement field can be recovered from the Ces\`{a}ro-Volterra formula:
\begin{align}\label{cesaro-volterra}
{u}_{i}({x})=\int_{{\gamma}({x})}\left ( E_{ij}(y)+\left ( \partial_{k}E_{ij}(y) - \partial_{i} E_{kj}(y)\right )\left ({x}_{k}-{y}_{k}\right )  \right ) \cdot d{y}_{j},
\end{align}
or equivalently in the vector form
\begin{align*}
{u}({x})=\int_{{\gamma}({x})} E({y})+({x}-{y})\wedge \left ( \nabla\times E({y})\right )\cdot  d{y},
\end{align*}
where ${\gamma}({x})$ is any smooth path connecting ${x}$ to a fixed point ${x}_{0}$. 
The derivative term $\nabla\times E$ appearing in the Ces\`{a}ro-Volterra path integral \eqref{cesaro-volterra} is called the Frank tensor (c.f. \cite{van2012distributional,van2016frank}). On simply connected domains, the integral \eqref{cesaro-volterra} does not depend on the chosen path between fixed end points.

We note that there has been a lot of recent progress and applications of the  Ces\`{a}ro-Volterra path integral. Non-simply-connected bodies are considered in \cite{yavari2013compatibility}. A generalization to weaker regularity is given in \cite{ciarlet2010cesaro} and a generalization to surfaces is given in \cite{ciarlet2009cesaro}. Geometric reductions for plate models are derived in \cite{geymonat2007kinematics} based on asymptotic expansions of the Ces\`{a}ro-Volterra integral. A compatible-incompatible decomposition of symmetric tensors in $L^p$ is proved in \cite{maggiani2015comp}.
The intrinsic elasticity models use the strain tensor as the major variable, and the displacement can be recovered by the Ces\`aro-Volterra path integral \cite{ciarlet2007characterization,ciarlet2009intrinsic}. The Frank tensor appearing in \eqref{cesaro-volterra} can be used as a boundary term \cite{van2016frank,ciarlet2009intrinsic}.

The rest of the paper will be organized as follows. In Section \ref{sec:preliminary} we define the notation and recall the Poincar\'e path integrals for the de Rham complex. In Section \ref{sec:results} we present the new Poincar\'{e} and Koszul operators for the elasticity complex. In Section \ref{sec:BGG} we review the derivation of the elasticity complex from the de Rham complex via the BGG construction. In Section \ref{sec:derivation} we propose a new methodology to derive Poincar\'{e} operators based on the  BGG construction and derive the operators for the elasticity complex. In Section \ref{sec:2D} we present results for the 2D elasticity complex.  Concluding remarks are given in Section \ref{sec:conclusion}.

\section{Notation and Preliminaries}\label{sec:preliminary}
Let $\mathbb{V}:=\mathbb{R}^n$ denote the space of vectors in ${\mathbb{R}^{n}}$, $\mathbb{M}$ denote the space of $n\times n$ matrices and $\mathbb{S}$, $\mathbb{K}$ for the subspaces of symmetric and skew-symmetric matrices respectively. We further define the product space $\mathbb{W}:=\mathbb{K}\times \mathbb{V}$. Let $\Lambda^{k}(\Omega)$ be the space of smooth $k$-forms on $\Omega$, and $\Lambda^{k}( \Omega; \mathbb{E})$ be the space of smooth $\mathbb{E}$-valued $k$-forms, where $\mathbb{E}=\mathbb{V}, \mathbb{M}, \mathbb{S}, \mathbb{K}$ or $\mathbb{W}$.   Similar notations $C^{\infty}(\Omega)$ and $C^{\infty}(\Omega, \mathbb{E})$  are used to denote smooth functions and $\mathbb{E}$-valued smooth functions on $\Omega$ respectively. When it is clear from the context, we also omit $\Omega$ and simply write $\Lambda^{k}(\mathbb{E})$ or $C^{\infty}(\mathbb{E})$.  We use lower case Latin letters for vector valued functions and upper case Latin letters for matrix valued functions. Greek letters are used for forms.


We define $\mathcal{P}_{r}(\mathbb{M})$ to be the space of matrix valued polynomials of degree at most $r$, and define $\mathcal{H}_{r}(\mathbb{M})$ to be the subspace of homogeneous polynomials of degree $r$, i.e. $Q\in \mathcal{H}_{r}(\mathbb{M})$ implies $Q(tx)=t^{r}Q(x)$ for any $t\in \mathbb{R}$. We also define similar spaces for symmetric matrices $\mathbb{S}$ and skew-symmetric matrices $\mathbb{K}$.

The notation $\nabla\times $ denotes the curl operator. For $W\in C^{\infty}(\mathbb{M})$, it is important to distinguish curl operators acting on the left and on the right:  $\nabla\times W$ is defined to be the curl applied to each column while $W\times \nabla$ is the curl applied to each row. Using index notation, this means $(\nabla\times W)_{ij}=\epsilon_{i}^{~ab}\partial_{a}W_{bj}$ and 
$(W\times \nabla)_{ij}=\epsilon_{j}^{~ab}\partial_{a}W_{ib}$ where the Einstein summation convention has been used. As a standard notation, we use ${u}\otimes {v}$ to denote the tensor product of the vectors ${u}$ and ${v}$, i.e. $({u}\otimes {v})_{ij}={u}_{i}{v}_{j}$.
Similarly, for a matrix $W$ and a vector ${u}$ we will use ${u}\wedge W$ to denote the cross product from the left, meaning the cross product between ${u}$ and the columns of $W$ (which returns a matrix), and $W\wedge {u}$ to denote the cross product from the right, i.e., between the rows of $W$ and the vector ${u}$.

Let $\mathrm{i}_{{v}}: \Lambda^{k}(\Omega)\mapsto \Lambda^{k-1}(\Omega)$ be the contraction operator with respect to a vector field ${v}$, defined by
$$
\mathrm{i}_{{v}}\omega(\xi_2 \ldots, \xi_k) := \omega ({v}, \xi_2, \ldots, \xi_k), \quad \omega\in \Lambda^{k}(\Omega). 
$$
In $\mathbb{R}^{n}$, we use ${x}$ to denote the identity vector field.  The Poincar\'e operator  $\mathfrak{p}$ with respect to the origin can be defined explicitly on $k$-forms $\omega$ by:
\begin{equation}\label{eq:poinexpl}
(\mathfrak{p}_{k} \omega)_x(\xi_2 \ldots, \xi_k) := \int_0^1 t^{k-1}\left ( \mathrm{i}_{{x}}\omega\right )_{ tx} (\xi_2, \ldots, \xi_k)\, d t= \int_0^1 t^{k-1}\omega_{ tx} ({x}, \xi_2, \ldots, \xi_k)\, d t.
\end{equation}
In vector form, the 3D Poincar\'e operators read:
\begin{equation*}
\begin{aligned}
\mathfrak{p}_{1}{u}&=\int_{0}^{1}{u}_{tx}\cdot {x}\,dt, \quad &&\forall u\in C^{\infty}(\mathbb{V}),
\\
\mathfrak{p}_{2}{v}&=\int_{0}^{1}t{v}_{tx}\wedge {x}\,dt, \quad &&\forall v\in C^{\infty}(\mathbb{V}),
\\
\mathfrak{p}_{3}w&=\int_{0}^{1}t^{2}w_{tx} {x}\,dt, \quad &&\forall w\in C^{\infty}(\mathbb{R}),
\end{aligned}
\end{equation*}
which satisfy
\begin{equation}\label{deRham1}
\begin{aligned}
\mathfrak{p}_{1}\grad f&=f+C, \quad &&\forall f\in C^{\infty}(\mathbb{R}),
\\
\mathfrak{p}_{2}\curl {u}+\grad\mathfrak{p}_{1}{u}&={u}, \quad &&\forall {u}\in C^{\infty}(\mathbb{V}),
\\
\mathfrak{p}_{3}\div {v}+\curl\mathfrak{p}_{2}{v}&={v}, \quad &&\forall {v}\in C^{\infty}(\mathbb{V}),
\\
\div\mathfrak{p}_{3} w &= w, \quad &&\forall w\in C^{\infty}(\mathbb{R}).
\end{aligned}
\end{equation}
Here, the notation $u_{x}$ is used to denote $u$ evaluated at $x$, i.e., $u(x)$.
In \eqref{deRham1} $C=-f(0)$ indicates that the identity $\mathfrak{p}_{1}\grad f= f$ holds up to a constant. We refer to \cite{ciarlet2013linear} for more details on the Poincar\'{e} operators for the de Rham complex and their relation to the Poincar\'{e} lemma. 

The contraction of a differential form by ${x}$ is called the Koszul operator (associated with the origin), i.e., 
\begin{equation}
\kappa_{k} \ : \  \omega \mapsto \kappa_{k}  \omega := {\mathrm{i}_{{x}}\omega}, \quad \omega\in \Lambda^{k}(\Omega). 
\end{equation}
The Koszul operators can be used to simplify the construction of some classical finite elements \cite[p. 29]{Arnold.D;Falk.R;Winther.R.2006a}.

Let $dx_{1}, dx_{2}, \cdots, dx_{n}$ be the canonical dual bases of $\mathbb{R}^{n}$. Then $dx_{\sigma_{0}}\wedge dx_{\sigma_{1}}\wedge \cdots \wedge dx_{\sigma_{k}}$ for all $0\leq \sigma_{0}<\sigma_{1}< \cdots < \sigma_{k}\leq n$ form a canonical basis for the vector space of alternating $k$-forms, i.e., any $\omega\in \Lambda^{k}$ can be written as
\begin{equation}\label{form-coefficient}
\omega=\sum_{0\leq \sigma_{1}< \cdots < \sigma_{k}\leq n}a_{\sigma}dx_{\sigma_{1}}\wedge dx_{\sigma_{2}}\wedge \cdots \wedge dx_{\sigma_{k}},
\end{equation}
for a unique choice of coefficients $a_{\sigma}\in \mathbb{R}$ \cite[p. 26]{Arnold.D;Falk.R;Winther.R.2006a}.

Let $\mathcal{H}_r\Lambda^k(\Omega)$ denote the space of $k$-forms with components ($a_{\sigma}$ in \eqref{form-coefficient}) that are homogeneous polynomials of degree $r$. Then from \eqref{eq:poinexpl} we have, for any $\omega\in\mathcal{H}_r\Lambda^k(\Omega)$, that
\begin{equation*}
\mathfrak{p}_{k} \omega = \frac{1}{ k + r} \kappa_{k} \omega\in \mathcal{H}_{r+1}\Lambda^{k-1}(\Omega).
\end{equation*}
Using the null-homotopy relation for Poincar\'e operators \eqref{eq:nullhom} we further have, for any $\omega\in\mathcal{H}_r\Lambda^{k}(\Omega)$, that
\begin{equation}
( d_{k-1}  \kappa_{k} + \kappa_{k+1}  d_{k} ) \omega = (r+k)\omega.
\end{equation}
Lastly, we remark that similar to the Poincar\'e operators, the Koszul operators also satisfy the complex property: $\kappa^2=0$.

The Poincar\'e and Koszul operators can also be defined with respect to another base point $x_{0}$, rather than the origin $0$. In this case, one replaces ${x}$ by ${x}-{x}_{0}$ in the contraction. To simplify the exposition we will in the remainder of this paper make the choice ${x}_{0}=0$ for the base point of our Poincar\'e and Koszul operators.




\section{Main results: Poincar\'{e} and Koszul operators for the elasticity complex}\label{sec:results}

In this section we state our main results for the elasticity complex. The proof of Theorem \ref{thm:1} is postponed to Section 5.  We remark again that the last two Poincar\'e operators in Theorem \ref{thm:1} are, as far as we know, new.%

\subsection{Poincar\'{e} operators}

\begin{theorem}\label{thm:1}
Let $\Omega=\mathbb{R}^{3}$ and let $\mathscr{P}_{1}: C^{\infty}(\Omega; \mathbb{S})\mapsto C^{\infty}(\Omega; \mathbb{V})$ be given by
$$
\mathscr{P}_{1} (E):= \int_{0}^{1}E_{tx}\cdot {x}\,dt +\int_{0}^{1}(1-t){x}\wedge (\nabla\times E_{tx})\cdot {x}\, dt,
$$
let $\mathscr{P}_{2}: C^{\infty}(\Omega; \mathbb{S})\mapsto C^{\infty}(\Omega; \mathbb{S})$ be given by
\begin{align*}
\mathscr{P}_{2}(V):= {x}\wedge \left (\int_{0}^{1}t(1-t)V_{tx}\, dt \right )\wedge {x},
\end{align*}
and let $\mathscr{P}_{3}: C^{\infty}(\Omega; \mathbb{V})\mapsto C^{\infty}(\Omega; \mathbb{S})$ be given by
\begin{align*}
\mathscr{P}_{3}({v}):=  \sym\left ( \int_{0}^{1}t^{2}{x}\otimes {v}_{tx}\,  dt -\left ( \int_{0}^{1}t^{2}(1-t){x}\otimes {v}_{tx}\wedge {x}\, dt\right )\times \nabla\right ).
\end{align*}
Then we have
\begin{equation}\label{RM}
\begin{aligned}
\mathscr{P}_{1}(\deff {u})&={u}+\mathrm{RM}, \quad&&\forall {u}\in C^{\infty}(\Omega; \mathbb{V}),
\\
\mathscr{P}_{2}\inc E + \deff \mathscr{P}_{1}E&=E, \quad &&\forall E\in C^{\infty}(\Omega; \mathbb{S}),
\\
\mathscr{P}_{3}\div V +\inc\mathscr{P}_{2}V&=V,\quad &&\forall V\in C^{\infty}(\Omega; \mathbb{S}),
\\
\div \mathscr{P}_{3}{v}&={v}, \quad&&\forall {v}\in C^{\infty}(\Omega; \mathbb{V}),
\end{aligned}
\end{equation}
where $\mathrm{RM}$ in \eqref{RM} indicates that the identity $\mathscr{P}_{1}(\deff {u})={u}$ holds up to rigid body motion, i.e. the kernel of $\deff$. 
Particularly, for a  symmetric matrix valued function $E$ satisfying $\inc E=0$, we have (the Ces\`{a}ro-Volterra path integral)
$$
E=\deff\left (  \mathscr{P}_{1}E\right ),
$$
and for a symmetric matrix valued function $V$ satisfying $\div V=0$, we have
$$
V=\inc\left ( \mathscr{P}_{2}V\right ).
$$
\end{theorem}
\begin{proof}
See Section 5.
\end{proof}
The above integrals are with respect to a special path ${\gamma}: t\mapsto tx$, connecting the base point $0$  with $x$. Since the Poincar\'{e} operators for the de Rham complex can be defined along an arbitrary path, we can also derive corresponding operators for the elasticity complex on a general path by following the BGG steps in the next two sections. Observe that by choosing the special path ${\gamma}: t\mapsto tx$ in the Ces\`ar{o}-Volterra formula \eqref{cesaro-volterra} we see that it coincides with the operator $\mathscr{P}_{1}$.

\begin{theorem}\label{thm:p2}
The Poincar\'{e} operators derived above satisfy the complex property $\mathscr{P}^2=0$.
\end{theorem}
\begin{proof}

We find from a straightforward calculation:
\begin{align}\label{curlR-3D}
({x}\otimes {v}\wedge {x})\times \nabla=3{x}\otimes {v} - {v}\otimes {x} + {x}\otimes \nabla_{{x}}{v}-(\nabla\cdot {v}){x}\otimes {x}.
\end{align}
Using \eqref{curlR-3D} and the fact that ${x}\wedge {x}=0$, we have the identity
$$
\mathscr{P}_{2}\mathscr{P}_{3}=0.
$$
Lastly, we find that
$$
\nabla\times ({x}\wedge V\wedge {x})=-3V\wedge {x}-({x}\cdot \nabla)V\wedge {x} +{x}\otimes \left (  (\div V)\wedge {x}\right )+{x}\otimes \operatorname{vec}\skw V,
$$
which implies that $\nabla\times ({x}\wedge V\wedge {x})\cdot {x}=0$ if $V$ is symmetric. Here $\operatorname{vec}: C^{\infty}(\mathbb{K})\mapsto C^{\infty}(\mathbb{V})$ is the canonical identification between a vector and a skew-symmetric matrix (see \eqref{Skw} where we define the inverse identification) and $\skw: C^{\infty}(\mathbb{M})\mapsto C^{\infty}(\mathbb{K})$ defined by $\skw V=1/2(V-V^{T})$ is the skew-symmetrization operator.  Therefore 
$$
\mathscr{P}_{1}\mathscr{P}_{2}=0.
$$
\end{proof}

\begin{theorem}
The Poincar\'{e} operators defined above are polynomial preserving: 
$$
E\in \mathcal{H}_{r}(\mathbb{S})\Rightarrow \mathscr{H}_{1}E\in \mathcal{H}_{r+1}(\mathbb{V}), 
\quad V\in \mathcal{H}_{r}(\mathbb{S})\Rightarrow \mathscr{H}_{2}V\in \mathcal{H}_{r+2}(\mathbb{S}), 
\quad {v}\in \mathcal{H}_{r}(\mathbb{V})\Rightarrow \mathscr{H}_{3}{v}\in \mathcal{H}_{r+1}(\mathbb{S}).
$$
\end{theorem}
Analogous to the de Rham case, the sequence 
\begin{align}\label{3DP-complex}
\begin{diagram}
0& \lTo^{} &\mathrm{RM} & \lTo^{} & C^{\infty}(\Omega; \mathbb{V})& \lTo^{\mathscr{P}_{1}} &  C^{\infty}(\Omega;  \mathbb{S}) & \lTo^{\mathscr{P}_{2}} &C^{\infty}(\Omega;  \mathbb{S}) & \lTo^{\mathscr{P}_{3}} &  C^{\infty}(\Omega; \mathbb{V}) & \lTo & 0
\end{diagram}
\end{align}
is a complex, since  $\mathscr{P}^{2}=0$. Furthermore, by the homotopy relation,  \eqref{3DP-complex} is exact if $\Omega$ is contractible.

\subsection{Koszul operators}

Analogous to the de Rham case we derive Koszul operators for the elasticity complex by applying the above Poincar\'{e} operators to homogeneous polynomials of degree $r$.

\begin{theorem}[Koszul operators]\label{thm:3D-Koszul}
Let the operator $\mathscr{K}_{1}^{r}: C^{\infty}(\mathbb{S})\mapsto C^{\infty}(\mathbb{V})$ be given by
$$
\mathscr{K}_{1}^{r}(E):= \frac{1}{r+1}E\cdot {x}+\frac{1}{(r+1)(r+2)}{x}\wedge (\nabla\times E) \cdot {x},
$$
the operator $\mathscr{K}_{2}^{r}: C^{\infty}(\mathbb{S})\mapsto C^{\infty}(\mathbb{S})$ be given by
$$
\mathscr{K}_{2}^{r}(V):= \frac{1}{(r+2)(r+3)}{x}\wedge V \wedge  {x},
$$
and the operator $\mathscr{K}_{3}^{r}: C^{\infty}(\mathbb{V})\mapsto C^{\infty}(\mathbb{S})$ be given by
$$
\mathscr{K}_{3}^{r}({v}):= \frac{1}{r+3}\sym ({x}\otimes {v}) -\frac{1}{(r+3)(r+4)}\sym \left ( ({x}\otimes {v}\wedge {x})\times \nabla\right ).
$$
Then $\mathscr{K}_{i}^{r}$, $i=1,2,3$, are the Koszul operators for the 3D elasticity complex.
\end{theorem}

As corollaries of the properties of the Poincar\'{e} operators (Theorem \ref{thm:p2}), the Koszul operators satisfy the homotopy identity, the complex property and the polynomial-preserving property on spaces of matrices and vectors whose components are homogeneous polynomials of degree $r$.
\begin{corollary}
For the Koszul operators, we have the homotopy identities
\begin{equation*}
\begin{aligned}
\mathscr{K}^{r-1}_{1}\deff {u}&={u}+\mathrm{RM},\quad&&\forall {u}\in \mathcal{H}_{r}(\mathbb{V}),
\\
\deff\mathscr{K}^{r}_{1} E+\mathscr{K}^{r-2}_{2}\inc E&=\omega,\quad&&\forall E\in \mathcal{H}_{r}(\mathbb{S}),
\\
\inc{\mathscr{K}}^{r}_{2}V+\mathscr{K}^{r-1}_{3}\div V&= V,\quad&&\forall V\in \mathcal{H}_{r}(\mathbb{S}),
\\
\div\mathscr{K}^{r}_{3} {v}&={v}, \quad&&\forall {v}\in \mathcal{H}_{r}(\mathbb{V}).
\end{aligned}
\end{equation*}
We have the  complex property $\mathscr{K}^{2}=0$, i.e.,
$$
\mathscr{K}^{r+2}_{1} \mathscr{K}^{r}_{2}=0, \quad\forall r=0, 1, \cdots,
$$
and
$$
\mathscr{K}^{r+1}_{2} \mathscr{K}^{r}_{3}=0, \quad\forall r=0, 1, \cdots.
$$
We also have the polynomial-preserving property:
$$
E\in \mathcal{H}_{r}(\mathbb{S})\Rightarrow \mathscr{K}_{1}^{r}E\in \mathcal{H}_{r+1}(\mathbb{R}), 
\quad V\in \mathcal{H}_{r}(\mathbb{S})\Rightarrow \mathscr{K}_{2}^{r}V\in \mathcal{H}_{r+2}(\mathbb{S}), 
\quad {v}\in \mathcal{H}_{r}(\mathbb{V})\Rightarrow \mathscr{K}^{r}_{3}{v}\in \mathcal{H}_{r+1}(\mathbb{S}).
$$
\end{corollary}
As a result of the above corollary, the sequence
\begin{align}\label{3DH-complex}
\begin{diagram}
0& \lTo^{} &\mathrm{RM} & \lTo^{} &  \mathcal{H}_{r+4}(\Omega; \mathbb{V})& \lTo^{\mathscr{K}_{1}^{r+3}} &   \mathcal{H}_{r+3}(\Omega;  \mathbb{S}) & \lTo^{\mathscr{K}_{2}^{r+1}} & \mathcal{H}_{r+1}(\Omega;  \mathbb{S}) & \lTo^{\mathscr{K}_{3}^{r}} &   \mathcal{H}_{r}(\Omega; \mathbb{V}) & \lTo & 0
\end{diagram}
\end{align}
is a complex. By the homotopy relation, it is exact if $\Omega$ is contractible.

\begin{remark}
Compared with the Koszul operators for the de Rham complexes, the definitions given in Theorem \ref{thm:3D-Koszul} contain correction terms involving derivatives ($\curl$ operators in 3D) and these terms explicitly depend on the degree of the homogeneous polynomials. As we have seen, these terms together with the polynomial degree naturally match with each other in the null-homotopy formulas and in the duality.  As discussed in the introduction for the classical Cesar\`o-Volterra formula, these derivative terms have physical significance in materials with incompatibility.
\end{remark}

\begin{remark}
Analogous to the de Rham case, we observe a relation of duality for the Koszul operators derived above.
 Specifically, let $\Omega$ be a star-shaped domain with respect to the origin.  We have for any $E, V\in C^{\infty}(\Omega; \mathbb{S})$:
$$
{x}\wedge E\wedge {x}: V=E: {x}\wedge V\wedge {x},
$$
which implies that
$$
\mathscr{K}_{2}^{r}E:V=E:\mathscr{K}_{2}^{r}V.
$$
Here $E:V$ denotes the Frobenius inner product of matrices $E$ and $V$.  
Moreover, for any $v\in C^{\infty}(\Omega; \mathbb{V})$, 
$$
\mathscr{K}_{1}^{r}E\cdot {v}={v}\cdot E\cdot {x}+\frac{1}{r+2}{v}\cdot ({x}\wedge \nabla\times {u})\cdot {x},
$$
and
$$
V:\mathscr{K}_{3}^{r}{v}=V:\sym ({x}\otimes {v})-\frac{1}{r+4}V:({x}\otimes {v}\wedge {x})\times \nabla.
$$
We note the identities
$$
V:\sym ({x}\otimes {v})={v}\cdot V\cdot {x},
$$
and the integration by parts
\begin{align*}
\int_{\Omega} V: \left [({x}\otimes {v}\wedge {x})\times \nabla\right ]
&=\int_{\Omega} -(V\times \nabla): ({x}\otimes {v}\wedge {x})=
-\int_{\Omega} {x}\cdot (V\times \nabla)\cdot ({v}\wedge {x})\\&
=\int_{\Omega} {x}\cdot (V\times \nabla)\wedge {x}\cdot {v}=-\int_{\Omega} {v}\cdot ({x}\wedge (\nabla\times  V))\cdot {x},
\end{align*}
which holds for $V$ with certain vanishing conditions on the boundary of the domain, e.g. $V\in C_{0}^{\infty}(\Omega; \mathbb{S})$.
This implies that
$$
\int_{\Omega} \mathscr{K}_{1}^{r+2}V\cdot {v}=\int_{\Omega} V:\mathscr{K}_{3}^{r}{v}.
$$
\end{remark}


\section{Bernstein-Gelfand-Gelfand Construction}\label{sec:BGG}
In this section, we recall the derivation of the elasticity complex from the de Rham complex by the  Bernstein-Gelfand-Gelfand (BGG) construction.  This will provide the preparation for the proof of Theorem \ref{thm:1}, which is given in Section \ref{sec:derivation}.
The BGG construction was originally developed in the theory of Lie algebras \cite{bernstein1975differential,vcap2001bernstein}. Later, Eastwood  \cite{eastwood2000complex,eastwood1999variations} showed the relation between  the elasticity complex and BGG. Arnold, Falk and Winther used the BGG construction to create finite element methods for elasticity \cite{arnold2006defferential,Arnold.D;Falk.R;Winther.R.2006a, falk2008finite}. 
The BGG construction for the 3D elasticity complex can be summarized in \eqref{BGG-3Delasticity}.

\begin{equation}\label{BGG-3Delasticity}
 \begin{tikzpicture}[node distance=3cm]
\node (L11) {$\Lambda^0(\mathbb{W})$};
\node[left of=L11] (L10) {$\mathbb{W}$};
\node[right of=L11] (L12) {$\Lambda^1(\mathbb{W})$};
\node[right of=L12] (L13) {$\Lambda^2(\mathbb{W})$};
\node[right of=L13] (L14) {$\Lambda^3(\mathbb{W})$};
\node[right of=L14] (L15) {$0$};
\node[below of=L11] (L21) {$\Lambda^0(\mathbb{W})$};
\node[left of =L21] (L20) {$\mathbb{W}$};
\node[right of=L21] (L22) {$\Gamma^1$};
\node[right of=L22] (L23) {$\Gamma^2$};
\node[right of=L23] (L24) {$\Lambda^3(\mathbb{W})$};
\node[right of=L24] (L25) {$0$};

\node[below of=L21] (L31) {$\Lambda^0(\mathbb{W})$};
\node[left of =L31] (L30) {$\mathbb{W}$};
\node[right of=L31] (L32) {$\Lambda^1(\mathbb{K})$};
\node[right of=L32] (L33) {$\Lambda^2(\mathbb{V})$};
\node[right of=L33] (L34) {$\Lambda^3(\mathbb{W})$};
\node[right of=L34] (L35) {$0$};

\node[below of=L31] (L41) {$C^{\infty}(\mathbb{V}\times \mathbb{K})$};
\node[left of =L41] (L40) {$\mathbb{V}\times \mathbb{K}$};
\node[right of=L41] (L42) {$C^{\infty}(\mathbb{M})$};
\node[right of=L42] (L43) {$C^{\infty}(\mathbb{M})$};
\node[right of=L43] (L44) {$C^{\infty}(\mathbb{K}\times \mathbb{V})$};
\node[right of=L44] (L45) {$0$};

\node[below of=L41] (L51) {$C^{\infty}(\mathbb{V})$};
\node[left of =L51] (L50) {$\mathbb{V}\times \mathbb{K}$};
\node[right of=L51] (L52) {$C^{\infty}(\mathbb{S})$};
\node[right of=L52] (L53) {$C^{\infty}(\mathbb{S})$};
\node[right of=L53] (L54) {$C^{\infty}(\mathbb{V})$};
\node[right of=L54] (L55) {$0$};

\draw[->] (L10)--(L11);
\draw[->] (L11)--(L12)  node[pos=.5,above] {$\mathscr{A}_{0}$};
\draw[->] (L12)--(L13)  node[pos=.5,above] {$\mathscr{A}_{1}$};
\draw[->] (L13)--(L14)  node[pos=.5,above] {$\mathscr{A}_{2}$};
\draw[->] (L14)--(L15);
\draw[->] (L20)--(L21);
\draw[->] (L21)--(L22)  node[pos=.5,above] {$\mathscr{A}_{0}$};
\draw[->] (L22)--(L23)  node[pos=.5,above] {$\mathscr{A}_{1}$};
\draw[->] (L23)--(L24)  node[pos=.5,above] {$\mathscr{A}_{2}$};
\draw[->] (L24)--(L25);
\draw[->] (L11)--(L21) node[pos=.5,right] {$\mathrm{id}$};
\draw[->] (L12)--(L22);
\draw[->] (L13)--(L23);
\draw[->] (L14)--(L24)node[pos=.5,right] {$\mathrm{id}$};

\draw[->] (L30)--(L31);
\draw[->] (L31)--(L32)  node[pos=.5,above] {$(d_{0}, -S_{0})$};
\draw[->] (L32)--(L33)  node[pos=.5,above] {$d_{1}S_{1}^{-1}d_{1}$};
\draw[->] (L33)--(L34)  node[pos=.5,above] {$(-S_{2}, d_{2})^{T}$};
\draw[->] (L34)--(L35);

\draw[->] (L40)--(L41);
\draw[->] (L41)--(L42)  node[pos=.5,above] {$(\grad, \mathrm{id})$};
\draw[->] (L42)--(L43)  node[pos=.5,above] {$\curl S_{1}^{-1}\curl$};
\draw[->] (L43)--(L44)  node[pos=.5,above] {$(\skw, \div)$};
\draw[->] (L44)--(L45);

\draw[->] (L50)--(L51);
\draw[->] (L51)--(L52)  node[pos=.5,above] {$\deff$};
\draw[->] (L52)--(L53)  node[pos=.5,above] {$\inc$};
\draw[->] (L53)--(L54)  node[pos=.5,above] {$\div$};
\draw[->] (L54)--(L55);

\draw[->] (L21)--(L31)node[pos=.5,right] {$\mathrm{id}$};
\draw[->] (L22)--(L32);
\draw[->] (L23)--(L33);
\draw[->](L24)--(L34)node[pos=.5,right] {$\mathrm{id}$};

\draw[->] (L31)--(L41);
\draw[->] (L32)--(L42) node[pos=.5,right] {$J_{1}$};
\draw[->] (L33)--(L43)node[pos=.5,right]{$J_{2}$};
\draw[->](L34)--(L44);

\draw[->] (L41)--(L51);
\draw[->] (L42)--(L52) node[pos=.5,right] {$\mathrm{sym}$};
\draw[->] (L43)--(L53)node[pos=.5,right]{$\mathrm{sym}$};
\draw[->](L44)--(L54);
{
\path (L12)--(L22) 
	node[pos=0.2,right] {$~~(\omega,\mu)$}
	node[pos=0.4,right]	{\quad~~$\downarrow$}
	node[pos=0.6,right] {$(\omega, S_{1}^{-1}d_{1}\omega)$} ;}
\path (L13)--(L23) 
	node[pos=0.2,right] {$~~(\omega,\mu)$}
	node[pos=0.4,right]	{~~\quad$\downarrow$}
	node[pos=0.6,right] {$(0, \mu+d_{1}S_{1}^{-1}\omega)$} ;
\path (L22)--(L32) 
	node[pos=0.2,right] {$(\omega, S_{1}^{-1}d_{1}\omega)$}
	node[pos=0.4,right]	{~~\quad$\downarrow$}
	node[pos=0.6,right] {~~\quad$\omega$} ;
\path (L23)--(L33) 
	node[pos=0.2,right] {$(0, \mu)$}
	node[pos=0.4,right]	{~~~$\downarrow$}
	node[pos=0.6,right] {~~~$\mu$} ;
\path (L31)--(L41) 
	node[pos=0.5,right] {$J_{0}$};
\path (L34)--(L44) 
    	node[pos=0.5,right]	{$J_{3}$};
\path (L44)--(L54) 
	node[pos=0.2,right] {$(W, {u})$}
	node[pos=0.4,right]	{~~~$\downarrow$}
	node[pos=0.6,right] {${u}-\div W$} ;
\path (L41)--(L51)
	node[pos=0.2,right] {$({u}, W)$}
	node[pos=0.4,right]	{~~~$\downarrow$}
	node[pos=0.6,right] {~~~${u}$} ;
\end{tikzpicture}
\end{equation}

The starting point of the BGG construction is the $\mathbb{W}$-valued de Rham complex
\begin{equation}
\begin{CD}
\mathbb{W}@>>>{\Lambda}^{0}(\mathbb{W}) @>d_{0}>> {\Lambda}^{1}(\mathbb{W})  @>d_{1}>>{\Lambda}^{0}(\mathbb{W}) @>d_{2} >>{\Lambda}^{3}(\mathbb{W})@ > >>  0,
\end{CD}
\end{equation}
where $\mathbb{W}:=\mathbb{K}\times \mathbb{V}$. Each element in ${\Lambda}^{k}(\mathbb{W})$ has two components, one is a skew-symmetric valued  $k$-form and another is a vector valued $k$-form.

Define $\mathscr{A}_{k}: {\Lambda}^{k}(\mathbb{W})\mapsto {\Lambda}^{k+1}( \mathbb{W})$ by 
$$
\mathscr{A}_{k}:=\left (
\begin{array}{cc}
d_{k} & -S_{k}\\
0 & d_{k}
\end{array}
\right ),
$$
i.e. for $(\omega, \mu)\in {\Lambda}^{k}(\mathbb{W})$, $\mathscr{A}_{k}(\omega, \mu):=(d_{k}\omega-S_{k}\mu, d_{k}\mu)$ with an operator $S_{k}: \Lambda^{k}(\mathbb{V})\mapsto \Lambda^{k+1}(\mathbb{K})$.
Before defining  $S_k$, we first introduce $K_{k}: \Lambda^{k}(\mathbb{V})\mapsto \Lambda^{k}(\mathbb{K})$ given by
$$
K_{k}(\omega):={x}\otimes \omega-\omega\otimes {x},
$$
where the definition is uniform with $k$.
The operator $K_{k}$ only acts on the coefficients of the alternating forms.  For any $k$-form, $K_{k}$ maps a vector coefficient to a skew-symmetric matrix coefficient.  

The operator $S_{k}$ is then defined by
$$
S_{k}:=d_{k}K_{k}-K_{k+1}d_{k}.
$$
By definition, the identity 
\begin{align}\label{dS}
d_{k+1}S_{k}+S_{k+1}d_{k}=0,
\end{align}
 holds.  From \eqref{dS} it is easy to see that $\mathscr{A}_{k+1}\mathscr{A}_{k}=0$, or $\mathscr{A}^{2}=0$ in short. 

It turns out that the operators $S_{k}$ are algebraic in the sense that no derivatives are involved.  Furthermore, $S_{0}$ is injective, $S_{1}$ is bijective and $S_{2}$ is surjective. 
 The first rows of  \eqref{BGG-3Delasticity}, i.e.
\begin{equation}
 \begin{diagram}
0& \rTo  & \mathbb{W}& \rTo  & \Lambda^{0}(\mathbb{W}) & \rTo^{\mathscr{A}_{0}} &\Lambda^{1}(\mathbb{W})& \rTo^{\mathscr{A}_{1}} &  \Lambda^{2}(\mathbb{W}) & \rTo^{\mathscr{A}_{2}} & \Lambda^{3}(\mathbb{W}) &\rTo & 0,
 \end{diagram}
 \end{equation}
is a complex. The second step is to filter out some parts of the complex.  The space $\Lambda^{2}(\mathbb{W})$ has two components, i.e. $\omega$ and $\mu$. In elasticity $\mu$ corresponds to the stress tensor, therefore we want to filter out the $\omega$ component.  Specifically, we consider the subspace of  $\Lambda^{2}(\mathbb{W})$: $\{(\omega, \mu)\in \Lambda^{2}(\mathbb{W}): \omega=0\}$. To find out the pre-image of the operator $\mathscr{A}_{2}$ restricted on this subspace, we notice that  $\mathscr{A}_{1}(u, v)=(d_{1}u-S_{1}v, d_{1}v)=(0, \mu)$ if and only if $v=S_{1}^{-1}d_{1}u$.  This gives the operators from the first row to the second:  we keep $\Lambda^{0}(\mathbb{W})$ and $\Lambda^{3}(\mathbb{W})$ and project $\Lambda^{1}(\mathbb{W})$ and $\Lambda^{2}(\mathbb{W})$ to the corresponding subspaces. One can verify that these operators are projections and the diagram commutes. 

The next step is an identification i.e. we identify $(\omega, S_{1}^{-1}d_{1}\omega)$ with $\omega$ and $(0, \mu)$ with $\mu$. This step also leads to a  commuting diagram. 

Then we identify differential forms and exterior derivatives with vectors/matrices and differential operators. Such identifications are called the vector proxies \cite[p. 26]{Arnold.D;Falk.R;Winther.R.2006a}. The operators $J_{k}$ provide vector/matrix representations of the differential forms. These representations are isomorphisms.  We will give explicit forms of the vector proxies below. This leads to the elasticity complex with weakly imposed symmetry (the fourth row of \eqref{BGG-3Delasticity}).

The last step is to project the complex into the subcomplex involving symmetric matrices.

\paragraph{Vector proxies in BGG}

We identify vector valued differential forms
$$
w=\left (
\begin{array}{c}
w_{1}\\
w_{2}\\
w_{3}
\end{array}
\right ) \sim \left (
\begin{array}{c}
w_{1}\\
w_{2}\\
w_{3}
\end{array}
\right ), \quad \forall w\in \Lambda^{0}(\mathbb{V});
$$
$$
w=\left (
\begin{array}{c}
w_{11}\\
w_{21}\\
w_{31}
\end{array}
\right )d x_{1} +
\left (
\begin{array}{c}
w_{12}\\
w_{22}\\
w_{32}
\end{array}
\right )d x_{2} 
+
\left (
\begin{array}{c}
w_{13}\\
w_{23}\\
w_{33}
\end{array}
\right )d x_{3} \sim 
\left (
\begin{array}{ccc}
w_{11} & w_{12} & w_{13}\\
w_{21} & w_{22} & w_{23}\\
w_{31} & w_{32} & w_{33}
\end{array}
\right ),
\quad \forall w\in \Lambda^{1}(\mathbb{V});
$$
$$
w=\left (
\begin{array}{c}
w_{11}\\
w_{21}\\
w_{31}
\end{array}
\right )d x_{2} \wedge d x_{3} +
\left (
\begin{array}{c}
w_{12}\\
w_{22}\\
w_{32}
\end{array}
\right )d x_{3} \wedge d x_{1} 
+
\left (
\begin{array}{c}
w_{13}\\
w_{23}\\
w_{33}
\end{array}
\right )d x_{1} \wedge d x_{2} \sim \left (
\begin{array}{ccc}
w_{11} & w_{12} & w_{13}\\
w_{21} & w_{22} & w_{23}\\
w_{31} & w_{32} & w_{33}
\end{array}
\right ), \quad \forall w\in \Lambda^{2}(\mathbb{V})
$$
$$
w=\left (
\begin{array}{c}
w_{1}\\
w_{2}\\
w_{3}
\end{array}
\right )d x_{1} \wedge d x_{2} \wedge d x_{3} \sim \left (
\begin{array}{c}
w_{1}\\
w_{2}\\
w_{3}
\end{array}
\right ), \quad \forall w\in \Lambda^{3}(\mathbb{V}).
$$
Here, a vector can be identified with a skew-symmetric matrix as
\begin{align}\label{Skw}
{w}=\left (
\begin{array}{c}
w_{1}\\
w_{2}\\
w_{3}
\end{array}\right )
\sim 
\mathrm{Skw}({w}):=\left (
\begin{array}{ccc}
0 & -w_{3} & w_{2}\\
w_{3} & 0 & -w_{1}\\
-w_{2} & w_{1} & 0
\end{array}
\right ).
\end{align}
Therefore the skew-symmetric matrix valued forms can be written as
$$
\mathrm{Skw}(w_{1}, w_{2}, w_{3}),
$$
$$
\mathrm{Skw}(w_{11}, w_{21}, w_{31})d x_{1} +\mathrm{Skw}(w_{12}, w_{22}, w_{32})d x_{2} +\mathrm{Skw}(w_{13}, w_{23}, w_{33})d x_{3} ,
$$
$$
\mathrm{Skw}(w_{11}, w_{21}, w_{31})d x_{2} \wedge d x_{3} +\mathrm{Skw}(w_{12}, w_{22}, w_{32})d x_{3} \wedge d x_{1} +\mathrm{Skw}(w_{13}, w_{23}, w_{33})d x_{1} \wedge d x_{2} ,
$$
$$
\mathrm{Skw}(w_{1}, w_{2}, w_{3})d x_{1} \wedge d x_{2} \wedge d x_{3} ,
$$
and the matrix proxies are obvious as discussed above.


The identifications $J_{0}: \Lambda^{0}(\mathbb{W})\mapsto C^{\infty}(\mathbb{V}\times \mathbb{K})$, $J_{1}: \Lambda^{1}(\mathbb{K})\mapsto C^{\infty}(\mathbb{M})$, $J_{2}: \Lambda^{2}(\mathbb{V})\mapsto C^{\infty}(\mathbb{M})$ and $J_{3}: \Lambda^{3}(\mathbb{W})\mapsto C^{\infty}(\mathbb{W})$ are defined by 
$$J_{0} (W, v):= (\Skw^{-1}W, \Skw v),$$
$$
J_{1}\left [ \mathrm{Skw}(w_{11}, w_{21}, w_{31})d x_{1} +\mathrm{Skw}(w_{12}, w_{22}, w_{32})d x_{2} +\mathrm{Skw}(w_{13}, w_{23}, w_{33})d x_{3} \right ]:=  \left (
\begin{array}{ccc}
w_{11} & w_{12} & w_{13}\\
w_{21} & w_{22} & w_{23}\\
w_{31} & w_{32} & w_{33}
\end{array}
\right ),
$$
$$
J_{2}\left [\left (
\begin{array}{c}
w_{11}\\
w_{21}\\
w_{31}
\end{array}
\right )d x_{2} \wedge d x_{3} +
\left (
\begin{array}{c}
w_{12}\\
w_{22}\\
w_{32}
\end{array}
\right )d x_{3} \wedge d x_{1} 
+
\left (
\begin{array}{c}
w_{13}\\
w_{23}\\
w_{33}
\end{array}
\right )d x_{1} \wedge d x_{2} \right ]:=   \left (
\begin{array}{ccc}
w_{11} & w_{12} & w_{13}\\
w_{21} & w_{22} & w_{23}\\
w_{31} & w_{32} & w_{33}
\end{array}
\right ),
$$
$$
J_{3} \left [ (W, v)dx_{1}\wedge dx_{2}\wedge dx_{3}\right ]:= (W, v).
$$

In the vector/matrix notation, the $S_{1}$ operator is of the form
$$
S_{1}W=W^{T}-\mathrm{tr}(W)I,
$$
and
$$
S_{1}^{-1}U=U^{T}-\frac{1}{2}\mathrm{tr}(U)I.
$$

\section{Derivation of the Poincar\'{e} operators}  \label{sec:derivation}
In this section we prove Theorem~\ref{thm:1}. We remark that our approach of using the BGG construction to derive explicit Poincar\'{e} path integrals for the elasticity complex is, as far as we know, a new methodology. The BGG construction is a general procedure for constructing differential complexes from the de Rham complex, and our results for the elasticity complex are thus a particular example of this approach.

\subsection{Poincar\'{e} operators on subcomplexes}

We first provide a general result for constructing null-homotopy operators on subcomplexes.
Assume that $W^{i}{\subseteq} V^{i}$, $(W, d)$ is a subcomplex of $(V, d)$ and the following diagram commutes, i.e. $\Pi_{i+1} d_{i}=d_{i}\Pi_{i}$,
\begin{equation}
 \begin{diagram}
 \cdots& \rTo  & V^{i-1} & \rTo^{d_{i-1}} &V^{i}& \rTo^{d_{i}} &  V^{i+1} & \rTo & \cdots\\
&&\dTo^{\Pi_{i-1}}  &   & \dTo^{\Pi_{i}}    & & \dTo^{\Pi_{i+1}} &\\
 \cdots& \rTo  & W^{i-1} & \rTo^{d_{i-1}} &W^{i}& \rTo^{d_{i}} &  W^{i+1} & \rTo & \cdots.
 \end{diagram}
 \end{equation}
We assume that $\Pi_{i}$ is surjective for each $i$, therefore there exists a right inverse of $\Pi_{i}$, which we denote as $\Pi_{i}^{\dagger}: W^{i}\mapsto V^{i}$. 

Suppose the top row has Poincar\'{e} operators $\mathfrak{p}_{i}: V^{i}\mapsto V^{i-1}$ satisfying
$$
\mathfrak{p}_{i+1}d_{i}+d_{i-1}\mathfrak{p}_{i}=\mathrm{id}_{V^{i}},
$$
then the next theorem shows how we can construct the Poincar\'{e} operator $\tilde{\mathfrak{p}}_{i}: W^{i}\mapsto W^{i-1}$ for the bottom row based on $\mathfrak{p}_{i}$ and the pseudo inverses of the operators $\Pi_{i}$. 
\begin{lemma}\label{lem:general-commuting}
If $\Pi^{\dagger}$ commutes with the differential operator $d$, i.e.
\begin{align}\label{commuting-dagger}
d_{i}\Pi^{\dagger}_{i}=\Pi^{\dagger}_{i+1}d_{i},
\end{align}
the formula
$$
\tilde{\mathfrak{p}}_{i}:=\Pi_{i-1}\mathfrak{p}_{i}\Pi_{i}^{\dagger}
$$
 defines an operator $\tilde{\mathfrak{p}}_{i}: W^{i}\mapsto W^{i-1}$ for the subcomplex $(W, d)$ satisfying 
 $$
 d_{i-1}\tilde{\mathfrak{p}}_{i}+\tilde{\mathfrak{p}}_{i+1}d_{i}=\mathrm{id}_{W^{i}}.
 $$
\end{lemma}
\begin{proof}
We have
$$
d_{i-1}\tilde{\mathfrak{p}}_{i}=\Pi_{i}d_{i-1}\mathfrak{p}_{i}\Pi_{i}^{\dagger},
$$
and
$$
\tilde{\mathfrak{p}}_{i+1}d_{i}=\Pi_{i}\mathfrak{p}_{i}d_{i}\Pi_{i}^{\dagger}.
$$
Therefore 
$$
d_{i-1}\tilde{\mathfrak{p}}_{i}+\tilde{\mathfrak{p}}_{i+1}d_{i}=\Pi_{i}\Pi_{i}^{\dagger}=\mathrm{id}_{W^{i}}.
$$
\end{proof}
If $\Pi_{j}$ is a projection, the inclusion operator $i: W^{j}\mapsto V^{j}$ naturally defines a right inverse of $\Pi_{j}$ and satisfies the commutative relation \eqref{commuting-dagger}.

\subsection{Poincar\'{e} operators on the elasticity complex}

The construction is summarized in \eqref{poincare-3Delasticity}.
\begin{equation}\label{poincare-3Delasticity}
 \begin{tikzpicture}[node distance=3cm]

\node (L11) {$\Lambda^0(\mathbb{W})$};
\node[left of=L11] (L10) {$\mathbb{W}$};
\node[right of=L11] (L12) {$\Lambda^1(\mathbb{W})$};
\node[right of=L12] (L13) {$\Lambda^2(\mathbb{W})$};
\node[right of=L13] (L14) {$\Lambda^3(\mathbb{W})$};
\node[right of=L14] (L15) {$0$};
\node[below of=L11] (L21) {$\Lambda^0(\mathbb{W})$};
\node[left of =L21] (L20) {$\mathbb{W}$};
\node[right of=L21] (L22) {$\Gamma^1$};
\node[right of=L22] (L23) {$\Gamma^2$};
\node[right of=L23] (L24) {$\Lambda^3(\mathbb{W})$};
\node[right of=L24] (L25) {$0$};

\node[below of=L21] (L31) {$\Lambda^0(\mathbb{W})$};
\node[left of =L31] (L30) {$\mathbb{W}$};
\node[right of=L31] (L32) {$\Lambda^1(\mathbb{K})$};
\node[right of=L32] (L33) {$\Lambda^2(\mathbb{V})$};
\node[right of=L33] (L34) {$\Lambda^3(\mathbb{W})$};
\node[right of=L34] (L35) {$0$};

\node[below of=L31] (L41) {$C^{\infty}(\mathbb{V}\times \mathbb{K})$};
\node[left of =L41] (L40) {$\mathbb{V}\times \mathbb{K}$};
\node[right of=L41] (L42) {$C^{\infty}(\mathbb{M})$};
\node[right of=L42] (L43) {$C^{\infty}(\mathbb{M})$};
\node[right of=L43] (L44) {$C^{\infty}(\mathbb{K}\times \mathbb{V})$};
\node[right of=L44] (L45) {$0$};

\node[below of=L41] (L51) {$C^{\infty}(\mathbb{V})$};
\node[left of =L51] (L50) {$\mathbb{V}\times \mathbb{K}$};
\node[right of=L51] (L52) {$C^{\infty}(\mathbb{S})$};
\node[right of=L52] (L53) {$C^{\infty}(\mathbb{S})$};
\node[right of=L53] (L54) {$C^{\infty}(\mathbb{V})$};
\node[right of=L54] (L55) {$0$};

\draw[<-] (L10)--(L11);
\draw[<-][pos=.5, below] (L11)--(L12)  node[pos=.5, above] {$\mathscr{B}_{1}$};
\draw[<-] (L12)--(L13)  node[pos=.5,above] {$\mathscr{B}_{2}$};
\draw[<-] (L13)--(L14)  node[pos=.5,above] {$\mathscr{B}_{3}$};
\draw[<-] (L14)--(L15);
\draw[<-] (L20)--(L21);
\draw[<-] (L21)--(L22)  node[pos=.5,above] {${\mathscr{C}_{1}}$};
\draw[<-] (L22)--(L23)  node[pos=.5,above] {${\mathscr{C}_{2}}$};
\draw[<-] (L23)--(L24)  node[pos=.5,above] {${\mathscr{C}_{3}}$};
\draw[<-] (L24)--(L25);
\draw[<-] (L11)--(L21) node[pos=.5,right] {$\mathrm{id}$};
\draw[<-] (L12)--(L22) node[pos=.5,right] {};
\draw[<-] (L13)--(L23) node[pos=.5,right] {};
\draw[<-] (L14)--(L24)node[pos=.5,right] {$\mathrm{id}$};

\draw[<-] (L30)--(L31);
\draw[<-] (L31)--(L32)  node[pos=.5,above] {$\mathscr{F}_{1}$};
\draw[<-] (L32)--(L33)  node[pos=.5,above] {$\mathscr{F}_{2}$};
\draw[<-] (L33)--(L34)  node[pos=.5,above] {$\mathscr{F}_{3}$};
\draw[<-] (L34)--(L35);

\draw[<-] (L40)--(L41);
\draw[<-] (L41)--(L42)  node[pos=.5,above] {$\tilde{\mathscr{F}}_{1}$};
\draw[<-] (L42)--(L43)  node[pos=.5,above] {$\tilde{\mathscr{F}}_{2}$};
\draw[<-] (L43)--(L44)  node[pos=.5,above] {$\tilde{\mathscr{F}}_{3}$};
\draw[<-] (L44)--(L45);

\draw[<-] (L50)--(L51);
\draw[<-] (L51)--(L52)  node[pos=.5,above] {$\mathscr{P}_{1}$};
\draw[<-] (L52)--(L53)  node[pos=.5,above] {$\mathscr{P}_{2}$};
\draw[<-] (L53)--(L54)  node[pos=.5,above] {$\mathscr{P}_{3}$};
\draw[<-] (L54)--(L55);

\draw[<-] (L21)--(L31)node[pos=.5,right] {$\mathrm{id}$};
\draw[<-] (L22)--(L32) node[pos=.5,right] {};
\draw[<-] (L23)--(L33) node[pos=.5,right] {};
\draw[<-](L24)--(L34)node[pos=.5,right] {$\mathrm{id}$};

\draw[<-] (L31)--(L41);
\draw[<-] (L32)--(L42) node[pos=.5,right] {};
\draw[<-] (L33)--(L43)node[pos=.5,right]{};
\draw[<-](L34)--(L44);

\draw[<-] (L41)--(L51);
\draw[<-] (L42)--(L52) node[pos=.5,right] {};
\draw[<-] (L43)--(L53)node[pos=.5,right]{};
\draw[<-](L44)--(L54);

{
\path (L12)--(L22) 
	node[pos=0.2,right] {$~~(\omega,\mu)$}
	node[pos=0.4,right]	{\quad~~$\downarrow$}
	node[pos=0.6,right] {$(\omega, S_{1}^{-1}d_{1}\omega)$} ;}
\path (L12)--(L22) 
	node[pos=0.2,left] {$~~(\omega, S_{1}^{-1}d_{1}\omega)$}
	node[pos=0.4,left]	{$\uparrow$\quad~~}
	node[pos=0.6,left] {$(\omega, S_{1}^{-1}d_{1}\omega)$} ;
\path (L13)--(L23) 
	node[pos=0.2,right] {$~~(\omega,\mu)$}
	node[pos=0.4,right]	{~~\quad$\downarrow$}
	node[pos=0.6,right] {$(0, \mu+d_{1}S_{1}^{-1}\omega)$} ;
\path (L13)--(L23) 	node[pos=0.2,left] {$~~(0,\mu)$}
	node[pos=0.4,left]	{$\uparrow$~~\quad}
	node[pos=0.6,left] {$(0, \mu)$} ;
\path (L22)--(L32) 
	node[pos=0.2,right] {$(\omega, S_{1}^{-1}d_{1}\omega)$}
	node[pos=0.4,right]	{~~\quad$\updownarrow$}
	node[pos=0.6,right] {~~\quad$\omega$} ;
\path (L23)--(L33) 
	node[pos=0.2,right] {$(0, \mu)$}
	node[pos=0.4,right]	{~~~$\updownarrow$}
	node[pos=0.6,right] {~~~$\mu$} ;
\path (L31)--(L41)
     node[pos=0.5,right]{$J_{0}^{-1}$};
\path (L32)--(L42)
     node[pos=0.5,right]{$J_{1}^{-1}$};
     \path (L33)--(L43)
     node[pos=0.5,right]{$J_{2}^{-1}$};
     \path (L34)--(L44)
     node[pos=0.5,right]{$J_{3}^{-1}$};
\path (L41)--(L51) 
	node[pos=0.2,right] {$({u}, W)$}
	node[pos=0.4,right]	{~~~$\downarrow$}
	node[pos=0.6,right] {~~~${u}$} ;
\path (L42)--(L52) 
	node[pos=0.2,right] {~~~${M}$}
	node[pos=0.4,right]	{~~~$\downarrow$}
	node[pos=0.6,right] {$\mathrm{sym}(M)$} ;
\path (L42)--(L52) 
	node[pos=0.2,left] {$V$~~~}
	node[pos=0.4,left]	{$\uparrow$~~~}
	node[pos=0.6,left] {$V$~~~} ;	
\path (L43)--(L53) 
	node[pos=0.2,right] {~~~$M$}
	node[pos=0.4,right]	{~~~$\downarrow$}
	node[pos=0.6,right] {$\mathrm{sym}(M)$} ;
\path (L43)--(L53) 
	node[pos=0.2,left] {$V$~~~}
	node[pos=0.4,left]	{$\uparrow$~~~}
	node[pos=0.6,left] {$V$~~~} ;	
\path (L44)--(L54) 
	node[pos=0.2,right] {$(W, {u})$}
	node[pos=0.4,right]	{~~~$\downarrow$}
	node[pos=0.6,right] {${u}-\div W$} ;
\path (L44)--(L54) 
	node[pos=0.2,left] {$(0, {u})$}
	node[pos=0.4,left]	{$\uparrow$~~~}
	node[pos=0.6,left] {${u}$~~~} ;	

\end{tikzpicture}
\end{equation}

The first step is to define an operator in the $\mathbb{W}$-valued de Rham complex $\mathscr{B}_{k}: \Lambda^{k}(\mathbb{W})\mapsto \Lambda^{k-1}(\mathbb{W})$ by 
\begin{align}
\mathscr{B}_{k}:=\left (
\begin{array}{cc}
\mathfrak{p}_{k} & -T_{k}\\
0 & \mathfrak{p}_{k}
\end{array}
\right ).
\end{align}
Here $T_{k}: \Lambda^{k}(\mathbb{V})\mapsto \Lambda^{k-1}(\mathbb{K})$ plays a similar role to $S_{k}$, but with the opposite direction ($S_{k}$ has degree 1 while $T_{k}$ has degree $-1$). We define $T_{k}$ as
\begin{align}\label{T-operator}
T_{k}:=\mathfrak{p}_{k}K_{k}-K_{k-1}\mathfrak{p}_{k}.
\end{align}
By straightforward calculations, we can check that 
\begin{align}\label{homotopy-AB}
\mathscr{A}_{k-1}\mathscr{B}_{k}+\mathscr{B}_{k+1}\mathscr{A}_{k}=\mathrm{id}_{\Lambda^{k}(\mathbb{W})}.
\end{align}
We have derived the homotopy inverses of the first row of \eqref{poincare-3Delasticity}. The next step is to perform several projections based on Lemma \ref{lem:general-commuting}.

To compute $\mathscr{C}_{3}$, we use the following path
\begin{equation}
 \begin{diagram}
(\mathfrak{p}_{3}\omega-T_{3}\mu, \mathfrak{p}_{3}\mu)& \lTo^{\mathscr{B}_{3}} & (\omega, \mu) \\
   \dTo&    & \uTo\\
\left (0, \mathfrak{p}_{3}\mu+d_{1}S_{1}^{-1}(\mathfrak{p}_{3}\omega-T_{3}\mu)\right )&  &  (\omega, \mu)
 \end{diagram}
 \end{equation}
 Therefore $\mathscr{C}_{3}$ maps $(\omega, \mu)$ to $\left (0, \mathfrak{p}_{3}\mu+d_{1}S_{1}^{-1}(\mathfrak{p}_{3}\omega-T_{3}\mu)\right )$. Similarly, we can obtain $\mathscr{C}_{2}$:
 \begin{equation}
 \begin{diagram}
(-T_{2}\mu, \mathfrak{p}_{2}\mu)& \lTo^{\mathscr{B}_{2}} & (0, \mu) \\
   \dTo&    & \uTo\\
\left (-T_{2}\mu, -S_{1}^{-1}d_{1}T_{2}\mu\right )&  &  (0, \mu)
 \end{diagram}
 \end{equation}
 For $\mathscr{C}_{1}$, we just have $\mathscr{C}_{1}=\mathscr{B}_{1}$.
 
Furthermore, for the third row we find $\mathscr{F}_{1}$, $\mathscr{F}_{2}$ and $\mathscr{F}_{3}$ by:
 \begin{equation}
 \begin{diagram}
\left (\mathfrak{p}_{1}\omega-T_{1}S_{1}^{-1}d_{1}\omega, \mathfrak{p}_{1}S_{1}^{-1}d_{1}\omega\right )& \lTo & (\omega, S_{1}^{-1}d_{1}\omega) \\
   \dTo&    & \uTo\\
\left (\mathfrak{p}_{1}\omega-T_{1}S_{1}^{-1}d_{1}\omega, \mathfrak{p}_{1}S_{1}^{-1}d_{1}\omega\right )&  &  \omega
 \end{diagram}
 \end{equation}
 \begin{equation}
 \begin{diagram}
\left (-T_{2}\mu, -S_{1}^{-1}d_{1}T_{2}\mu\right )& \lTo & (0, \mu) \\
   \dTo&    & \uTo\\
-T_{2}\mu&  &  \mu
 \end{diagram}
 \end{equation}
\begin{equation}
 \begin{diagram}
\left (0, \mathfrak{p}_{3}\mu+d_{1}S_{1}^{-1}(\mathfrak{p}_{3}\omega-T_{3}\mu)\right )& \lTo & (\omega, \mu) \\
   \dTo&    & \uTo\\
\mathfrak{p}_{3}\mu+d_{1}S_{1}^{-1}(\mathfrak{p}_{3}\omega-T_{3}\mu)&  &  (\omega, \mu)
 \end{diagram}
 \end{equation}

The next step is to consider vector proxies given by $J_{0}$, $J_{1}$, $J_{2}$ and $J_{3}$.

\paragraph{Matrix proxy}

We give the vector-matrix forms of the above constructions. 
For $\mu\in \Lambda^{1}(\mathbb{V})$, 
we have
$$
T_{1}{\mu}_x\sim -\int_{0}^{1}(1-t){x}\wedge \left (J_{1}\mu\right )_{tx}\cdot {x}\,dt,
$$
and so for a matrix $M\in \mathbb{M}$, this gives
$$
\tilde{\mathscr{F}}_{1}: M\mapsto \left ( \int_{0}^{1}M_{tx}\cdot {x}\,dt +\int_{0}^{1}(1-t){x}\wedge  (M_{tx}\times \nabla)\cdot {x}\,dt,~  \int_{0}^{1}\left [ (M_{tx}\times \nabla)^{T}-\frac{1}{2}(M_{tx}\times \nabla)I\right ] \cdot {x}\, dt
\right ) .
$$

If $\mu\in \Lambda^{2}(\mathbb{V})$, then
$$
T_{2}\mu_x\sim \int_{0}^{1}t(1-t){x}\wedge (J_{2}\mu)_{tx}\wedge {x}\,dt,
$$
and so for $M\in \mathbb{M}$, we have
\begin{align}
\tilde{\mathscr{F}}_{2}: M\mapsto  \int_{0}^{1}t(1-t){x}\wedge M_{tx} \wedge {x} \, dt ={x}\wedge \left (\int_{0}^{1}t(1-t)M_{tx}\,dt \right )\wedge {x}.
\end{align}

Lastly, with $\mu\in \Lambda^{3}(\mathbb{V})$,  
$$
T_{3}\mu_x\sim -\int_{0}^{1}t^{2}(1-t){x}\wedge (J_{3}\mu)_{tx}\otimes {x}\, dt,
$$
and so for $(W, {v})\in C^{\infty}(\mathbb{K})\times C^{\infty}(\mathbb{V})$, this leads to 
\begin{align}
\tilde{\mathscr{F}}_{3}: (W, {v})\mapsto  \int_{0}^{1}t^{2}{v}_{tx}\otimes {x}\, dt +\left ( S_{1}^{-1}\int_{0}^{1}t^{2}(1-t){x}\wedge {v}_{tx}\otimes {x}\, dt\right )\times \nabla + \int_{0}^{1}t^{2}\left [ {x}\otimes W_{tx}-1/2({x}\cdot W_{tx})I\right ]\times \nabla\, dt.
\end{align}

Finally, we perform several symmetrizations to get the Poincar\'e operators for the elasticity complex.
If $E$ is a symmetric matrix, we have $\mathrm{tr}(E\times \nabla)=0$. Therefore 
$\mathscr{P}_{1}$ can be interpreted as
$$
\mathscr{P}_{1}: E\mapsto \int_{0}^{1}E_{tx}\cdot {x}\,dt +\int_{0}^{1}(1-t){x}\wedge (\nabla\times E_{tx})\cdot {x}\, dt,\quad E\in C^{\infty}(\mathbb{S}),
$$
which is the Ces\`{a}ro-Volterra formula.

Moreover, the vector proxy of $\mathscr{P}_{2}$ reads:
\begin{align}
\mathscr{P}_{2}: V\mapsto \mathrm{sym}T_{2}V=\mathrm{sym}\left ( \int_{0}^{1}t(1-t){x}\wedge V_{tx} \wedge {x}\, dt  \right )={x}\wedge \left (\int_{0}^{1}t(1-t)V_{tx}\,dt \right )\wedge {x}.
\end{align}
whenever $V$ is  symmetric.

For $\mathscr{P}_{3}$, we have
\begin{align}
\mathscr{P}_{3}: {v}\mapsto \mathrm{sym} (\mathfrak{p}_{3}{ {v}}+d_{1}S_{1}^{-1}T_{3}{ {v}})=\mathrm{sym}\left (  \int_{0}^{1}t^{2}{ {v}_{tx}}\otimes {x} \,dt +\left ( S_{1}^{-1}\int_{0}^{1}t^{2}(1-t){x}\wedge { {v}_{tx}}\otimes {x} \,dt\right )\times \nabla\right ),
\end{align}
where we recall that $S_{1}^{-1}M:=M^{T}-1/2\mathrm{tr}(M)$.  For any vector ${u}$, the matrix ${x}\wedge {u}\otimes {x}$ has the index form $\left ( {x}\wedge {u}\otimes {x}\right )_{il}=\epsilon_{ijk}{x}^{j}{u}^{k}{x}_{l}$, from which we can easily see that $\mathrm{tr}\left ({x}\wedge {u}\otimes {x}\right )=\epsilon_{ijk}{x}^{j}{u}^{k}{x}^{i}=0$. 
Therefore $\mathscr{P}_{3}$ is reduced to 
\begin{align}
\mathscr{P}_{3}: { {v}}\mapsto  \sym\left ( \int_{0}^{1}t^{2}{x}\otimes { {v}_{tx}} \,dt -\left ( \int_{0}^{1}t^{2}(1-t){x}\otimes { {v}_{tx}}\wedge{x} \,dt\right )\times \nabla\right ).
\end{align}

\section{2D elasticity complex}\label{sec:2D}
Let $\Omega$ be a contractible domain in 2D. The {elasticity complex} in 2D reads
\begin{equation}\label{2D-elasticity0}
 \begin{diagram}
{0}& \rTo^{} &{\mathcal{P}_1}& \rTo^{{\subseteq}} & C^{\infty}(\Omega)& \rTo^{\air} &  C^{\infty}(\Omega;  \mathbb{S}) & \rTo^{\div} &  C^{\infty}(\Omega; \mathbb{V}) & \rTo & 0.
 \end{diagram}
 \end{equation}
  The Airy operator, $\air: C^{\infty}(\mathbb{R})\mapsto C^{\infty}(\mathbb{S})$ is defined by $\air(u):= \nabla\times u\times \nabla$ in 2D, is a rotated version of the Hessian:
 $$
 \air (u):=\left ( 
 \begin{array}{cc}
{\partial_{2}^{2}u} &  -\partial_{1}\partial_{2} u \\
  -\partial_{1}\partial_{2} u &{\partial_{1}^{2}u}
 \end{array}
 \right ).
 $$
 In planar elasticity, the Cauchy stress is a second order symmetric tensor, appearing in  $ C^{\infty}(\Omega; \mathbb{S})$ in the sequence \eqref{2D-elasticity0}. The  divergence operator, defined row-wise, maps onto $ C^{\infty}$ vectors and the kernel can be parametrized by $ C^{\infty}$ scalar functions through the Airy operator.

For ${x}=(x_{1}, x_{2})$ we let ${x}^{\perp}=(x_{2}, -x_{1})$. We use
 $$
 {\chi}=
 \left (
 \begin{array}{cc} 
 0 & -1\\
 1 & 0
 \end{array}
 \right )
 $$
to denote the  canonical skew-symmetric matrix in 2D.

In the 2D case, we assume ${v}=({v}_{1}, {v}_{2})^{T}$. Then 
$$
K_{k}({v}):={x}\otimes {v}-{v}\otimes {x}=\left (
\begin{array}{cc}
0 & {x}_{1}{v}_{2}-{x}_{2}{v}_{1}\\
-\left ({x}_{1}{v}_{2}-{x}_{2}{v}_{1} \right ) & 0
\end{array}
\right ).
$$
This anti-symmetric matrix is usually identified with the scalar $-\left ({x}_{1}{v}_{2}-{x}_{2}{v}_{1} \right )$.


\begin{theorem}[2D case]\label{2D:poincare-results}
Assume $\Omega=\mathbb{R}^{2}$. We define $\mathscr{P}_{1}: C^{\infty}(\mathbb{S})\mapsto C^{\infty}(\mathbb{R})$ by
$$
 \mathscr{P}_{1}: V\mapsto \int_{0}^{1}(1-t){x}^{\perp}\cdot V_{tx}\cdot {x}^{\perp} \, dt.
$$
and define $\mathscr{P}_{2}: C^{\infty}(\mathbb{V})\mapsto C^{\infty}(\mathbb{S})$ by
$$
 \mathscr{P}_{2}: {u}\mapsto\sym\left ( \int_{0}^{1}t {u}_{tx}\otimes {x}\, dt+\left (\int_{0}^{1}t(t-1)({x}^{\perp}\cdot {u}_{tx}){x} \, dt\right )\times \nabla\right ),
$$ 
where the 2D scalar curl operator ``$\times \nabla$'' maps each component of the vector $\sym\int_{0}^{1}t {u}_{tx}\otimes {x}\, dt+\left (\int_{0}^{1}t(t-1)({x}^{\perp}\cdot {u}_{tx}){x} \, dt\right )$ to a row vector.
Then we have
\begin{align}\label{P1}
\mathscr{P}_{1}(\air u)=u+\mathcal{P}_{1}, \quad\forall u\in C^{\infty}(\Omega),
\end{align}
$$
\mathscr{P}_{2}\div V + \air \mathscr{P}_{1}V=V, \quad \forall V\in C^{\infty}(\Omega; \mathbb{S}),
$$
and
$$
\div \mathscr{P}_{2}{v}={v}, \quad\forall {v}\in C^{\infty}(\Omega; \mathbb{V}),
$$
where $\mathcal{P}_{1}$ in \eqref{P1} indicates that the identity $\mathscr{P}_{1}\air u=u$ holds up to  $\mathcal{P}_{1}$, the kernel of $\air$. 

Particularly, for a matrix field $V$ satisfying $\div V=0$, we can find a scalar function $f:=\mathscr{P}_{1}V$ satisfying $\air f=V$. For any vector function ${v}$, we can explicitly find a symmetric matrix potential $M:=\mathscr{P}_{2}{v}$ satisfying $\div M={v}$.
\end{theorem}

Similar to the 3D case \eqref{3DP-complex}, the 2D version
\begin{align}\label{2DP-complex}
\begin{diagram}
0& \lTo^{} &\mathrm{RM} & \lTo^{} & C^{\infty}(\Omega; \mathbb{V})& \lTo^{\mathscr{P}_{1}} &  C^{\infty}(\Omega;  \mathbb{S}) & \lTo^{\mathscr{P}_{2}} &  C^{\infty}(\Omega) & \lTo & 0
\end{diagram}
\end{align}
is also a complex, which is exact on contractible domains. Koszul operators can be similarly obtained.

The construction for the Poincar\'{e} operators for the 2D elasticity complex can be summarized in the diagram below.

\begin{equation}\label{poincare-2Delasticity}
 \begin{tikzpicture}[node distance=3cm]

\node (L11) {$\mathbb{W}$};
\node[right of=L11] (L12) {$\Lambda^0(\mathbb{W})$};
\node[right of=L12] (L13) {$\Lambda^1(\mathbb{W})$};
\node[right of=L13] (L14) {$\Lambda^2(\mathbb{W})$};
\node[right of=L14] (L15) {$0$};
\node[below of=L11] (L21) {$\mathbb{W}$};
\node[right of=L21] (L22) {$\Gamma^0$};
\node[right of=L22] (L23) {$\Gamma^1$};
\node[right of=L23] (L24) {$\Lambda^2(\mathbb{W})$};
\node[right of=L24] (L25) {$0$};

\node[below of=L21] (L31) {$\mathbb{W}$};
\node[right of=L31] (L32) {$\Lambda^{0}(\mathbb{K})$};
\node[right of=L32] (L33) {$\Lambda^{1}(\mathbb{V})$};
\node[right of=L33] (L34) {$\Lambda^{2}(\mathbb{W})$};
\node[right of=L34] (L35) {$0$};

\node[below of=L31] (L41) {$\mathbb{W}$};
\node[right of=L41] (L42) {$C^{\infty}(\mathbb{R})$};
\node[right of=L42] (L43) {$C^{\infty}(\mathbb{M})$};
\node[right of=L43] (L44) {$C^{\infty}(\mathbb{W})$};
\node[right of=L44] (L45) {$0$};

\node[below of=L41] (L51) {$\mathbb{W}$};
\node[right of=L51] (L52) {$C^{\infty}(\mathbb{R})$};
\node[right of=L52] (L53) {$C^{\infty}(\mathbb{S})$};
\node[right of=L53] (L54) {$C^{\infty}(\mathbb{V})$};
\node[right of=L54] (L55) {$0$};

\draw[<-][pos=.5, below] (L11)--(L12)  node[pos=.5, above] {};
\draw[<-] (L12)--(L13)  node[pos=.5,above] {$\mathscr{B}_{1}$};
\draw[<-] (L13)--(L14)  node[pos=.5,above] {$\mathscr{B}_{2}$};
\draw[<-] (L14)--(L15);
\draw[<-] (L21)--(L22)  node[pos=.5,above] {};
\draw[<-] (L22)--(L23)  node[pos=.5,above] {${\mathscr{C}_{1}}$};
\draw[<-] (L23)--(L24)  node[pos=.5,above] {${\mathscr{C}_{2}}$};
\draw[<-] (L24)--(L25);
\draw[<-] (L12)--(L22) node[pos=.5,right] {};
\draw[<-] (L13)--(L23) node[pos=.5,right] {};
\draw[<-] (L14)--(L24)node[pos=.5,right] {$\mathrm{id}$};

\draw[<-] (L31)--(L32)  node[pos=.5,above] {};
\draw[<-] (L32)--(L33)  node[pos=.5,above] {$\mathscr{F}_{1}$};
\draw[<-] (L33)--(L34)  node[pos=.5,above] {$\mathscr{F}_{2}$};
\draw[<-] (L34)--(L35);

\draw[<-] (L41)--(L42)  node[pos=.5,above] {};
\draw[<-] (L42)--(L43)  node[pos=.5,above] {$\tilde{\mathscr{F}}_{1}$};
\draw[<-] (L43)--(L44)  node[pos=.5,above] {$\tilde{\mathscr{F}}_{2}$};
\draw[<-] (L44)--(L45);

\draw[<-] (L51)--(L52)  node[pos=.5,above] {};
\draw[<-] (L52)--(L53)  node[pos=.5,above] {$\mathscr{P}_{1}$};
\draw[<-] (L53)--(L54)  node[pos=.5,above] {$\mathscr{P}_{2}$};
\draw[<-] (L54)--(L55);

\draw[<-] (L22)--(L32) node[pos=.5,right] {};
\draw[<-] (L23)--(L33) node[pos=.5,right] {};
\draw[<-](L24)--(L34)node[pos=.5,right] {$\mathrm{id}$};

\draw[<-] (L32)--(L42) node[pos=.5,right] {};
\draw[<-] (L33)--(L43)node[pos=.5,right]{};
\draw[<-](L34)--(L44);

\draw[<-] (L42)--(L52) node[pos=.5,right] {};
\draw[<-] (L43)--(L53)node[pos=.5,right]{};
\draw[<-](L44)--(L54);

{
\path (L12)--(L22) 
	node[pos=0.2,right] {$~~(\omega,\mu)$}
	node[pos=0.4,right]	{\quad~~$\downarrow$}
	node[pos=0.6,right] {$(\omega, S_{0}^{-1}d_{0}\omega)$} ;}
\path (L12)--(L22) 
	node[pos=0.2,left] {$~~(\omega, S_{0}^{-1}d_{0}\omega)$}
	node[pos=0.4,left]	{$\uparrow$\quad~~}
	node[pos=0.6,left] {$(\omega, S_{0}^{-1}d_{0}\omega)$} ;
\path (L13)--(L23) 
	node[pos=0.2,right] {$~~(\omega,\mu)$}
	node[pos=0.4,right]	{~~\quad$\downarrow$}
	node[pos=0.6,right] {$(0, \mu+d_{0}S_{0}^{-1}\omega)$} ;
\path (L13)--(L23) 	node[pos=0.2,left] {$~~(0,\mu)$}
	node[pos=0.4,left]	{$\uparrow$~~\quad}
	node[pos=0.6,left] {$(0, \mu)$} ;
\path (L22)--(L32) 
	node[pos=0.2,right] {$(\omega, S_{0}^{-1}d_{0}\omega)$}
	node[pos=0.4,right]	{~~\quad$\updownarrow$}
	node[pos=0.6,right] {~~\quad$\omega$} ;
\path (L23)--(L33) 
	node[pos=0.2,right] {$(0, \mu)$}
	node[pos=0.4,right]	{~~~$\updownarrow$}
	node[pos=0.6,right] {~~~$\mu$} ;
\path (L42)--(L52) 
	node[pos=0.2,right] {}
	node[pos=0.4,right]	{$\mathrm{id}$}
	node[pos=0.6,right] {} ;
\path (L43)--(L53) 
	node[pos=0.2,right] {~~~$M$}
	node[pos=0.4,right]	{~~~$\downarrow$}
	node[pos=0.6,right] {$\mathrm{sym}(M)$} ;
\path (L43)--(L53) 
	node[pos=0.2,left] {$V$~~~}
	node[pos=0.4,left]	{$\uparrow$~~~}
	node[pos=0.6,left] {$V$~~~} ;	
\path (L44)--(L54) 
	node[pos=0.2,right] {$(V, {v})$}
	node[pos=0.4,right]	{~~~$\downarrow$}
	node[pos=0.6,right] {${v}-\div V$} ;
\path (L44)--(L54) 
	node[pos=0.2,left] {$(0, {v})$}
	node[pos=0.4,left]	{$\uparrow$~~~}
	node[pos=0.6,left] {${v}$~~~} ;

\path (L32)--(L42) 
	node[pos=0.5,right] {$J_{0}^{-1}$};
\path (L33)--(L43) 
	node[pos=0.5,right] {$J_{1}^{-1}$};
\path (L34)--(L44) 
	node[pos=0.5,right] {$J_{2}^{-1}$};

\end{tikzpicture}
\end{equation}

The derivation for the 2D Poincar\'{e} operators is analogous to the 3D case discussed above. The results in Theorem \ref{2D:poincare-results} can thus be obtained in a similar manner. We omit the details. 


\section{Conclusion}\label{sec:conclusion}

In this paper, we derived the null-homotopy operators for the elasticity complex. By construction they automatically satisfy the homotopy relation $\mathscr{D}_{i-1}\mathscr{P}_{i}+\mathscr{P}_{i+1}\mathscr{D}_{i}=\mathrm{id}$. The complex property $\mathscr{P}^{2}=0$ and the polynomial-preserving property were also shown. As the de Rham case, for any $\omega\in V^{i}$ with $\mathscr{D}_{i}\omega=0$, a potential $\phi\in V^{i-1}$ satisfying $\mathscr{P}_{i-1}\phi=0$ and $\mathscr{D}_{i-1}\phi=\omega$ is uniquely determined, and is given by $\phi=\mathscr{P}_{i}\omega$.


As a special case, the classical Cesar\`{o}-Volterra path integral is derived from the first Poincar\'{e} operator for the de Rham complex. The known path independence of the Cesar\`{o}-Volterra integral can thus be seen as a corollary of the known path independence of the Poincar\'{e} operator for differential $1$-forms.

The method discussed in this paper would work for any complex obtained by the BGG construction. The elasticity complex is just one special case, and more examples can be found in, e.g., \cite{arnold2015beijing,eastwood1999variations}. Therefore, Poincar\'{e} operators for these complexes can be also constructed following an analogous approach.

As future work, we hope that the methodology and the results in this paper can be useful in establishing regularity results for the elasticity complex based on estimates of regularized integral operators (c.f., \cite{costabel2010bogovskiui}) and in the investigation of Poincar\'{e} operators on manifolds or with little regularity, as studied for the Ces\`ar{o}-Volterra formula \cite{ciarlet2010cesaro,ciarlet2009cesaro}.



\section*{Acknowledgement}

The authors are grateful to Douglas Arnold and Ragnar Winther for valuable feedback that helped improve the manuscript. 

The research of KH leading to the results of this paper was partly carried out during his affiliation with the University of Oslo.  KH and ES were supported in part by the  European Research Council under the European Union's Seventh Framework Programme (FP7/2007-2013) / ERC grant agreement 339643.

\bibliographystyle{plain}      
\bibliography{poincare}{}   

\end{document}